\documentclass[12pt]{article}
\usepackage{amsmath,amssymb,amsthm,enumerate}
\usepackage[dvipdfmx]{hyperref}
\usepackage{natbib}

\def\C{{\mathcal{C}}}
\def\E{{\mathcal{E}}}

\def\R{{\mathbb{R}}}

\def\U{{\mathcal{U}}}

\theoremstyle{plain}
\newtheorem{thm}{Theorem}[section]
\newtheorem{lem}{Lemma}[section]

\newtheorem{cor}{Corollary}[section]

\theoremstyle{definition}
\newtheorem{dfn}{Definition}[section]

\newtheorem{rem}{Remark}[section]

\numberwithin{equation}{section}
\allowdisplaybreaks

\title{Fuzzy Core Equivalence in Large Economies: A Role for the Infinite-\hspace{0pt}Dimensional Lyapunov Theorem\thanks{A preliminary version of this manuscript is an outgrowth of Sagara's visit at Johns Hopkins University during 2012--2014. This research is supported by JSPS KAKENHI Grant Number JP18K01518 from the Ministry of Education, Culture, Sports, Science and Technology, Japan.}}

\date{\today}
\author{M. Ali Khan\thanks{Corresponding author.}
\\
{\small Department of Economics, The Johns Hopkins University} \\[-2pt]
{\small Baltimore, MD 21218, United States} \\
{\small e-mail: akhan@jhu.edu}
\and
\\
Nobusumi Sagara  \\
{\small Department of Economics, Hosei University} \\[-2pt]
{\small 4342, Aihara, Machida, Tokyo 194--0298, Japan} \\
{\small e-mail: nsagara@hosei.ac.jp}}

\begin{document}
\begin{titlepage}
\maketitle

\setcounter{page}{0}
\thispagestyle{empty}
\clearpage

\pagestyle{empty}
\begin{abstract}
We present the equivalence between the fuzzy core and the core under minimal assumptions. Due to the exact version of the Lyapunov convexity theorem in Banach spaces, we clarify that the additional structure of commodity spaces and preferences is unnecessary whenever the measure space of agents is ``saturated''. As a spin-\hspace{0pt}off of the above equivalence, we obtain the coincidence of the core, the fuzzy core, and the Schmeidler's restricted core under minimal assumptions. The coincidence of the fuzzy core and the restricted core has not been articulated anywhere. \\[-6pt]

\noindent
{\bfseries Key words:} large economy; fuzzy core; core; restricted core; infinite-\hspace{0pt}dimensional commodity space; Lyapunov's theorem; saturated measure space. 
\\[-6pt]

\noindent
{\bfseries JEL Classification System:} D51 
\end{abstract}
\end{titlepage}
\tableofcontents

\section{Introduction}
Classical core theory in exchange economies deals with the situation where agents have only one of two alternative possibilities: whether to join or not to join a coalition  -- there is no option for varying degrees of commitment to, and participation in, more than a single one  coalition. On the contrary, fuzzy core theory proposed by \citet{au81,au82} allow for the partial participation of agents in coalitions where the attainable outcomes of the allocations of goods depend on the degree of commitment and participation of the agents, thereby modelling a more pluralistic conception of decision-\hspace{0pt}making and  identity formation as is evident from even the casual observations of group behavior.

In particular, the significant observation Aubin made is the equivalence between the fuzzy core and the set of Walrasian allocations (fuzzy core--Walras equivalence) in finite agent economies with finite-\hspace{0pt}dimensional commodity spaces; see \citet{hu94} for the case of large economies. This is another remarkable coincidence result on the core since the well-\hspace{0pt}known Debreu--\hspace{0pt}Scarf limit theorem for Edgeworth replica economies; see \citet{ds63}. 

Motivated by Aubin's formulation, \citet{no00,bg14} demonstrated the equivalence between the core and fuzzy core in atomless economies. Although the presence of price systems is irrelevant to the direct verification of such equivalence, they need to assume the infinite-\hspace{0pt}dimensional commodity space to be an ordered Banach space whose positive cone has an nonempty interior and preferences of each agent to be strictly monotone. This is a standard assumption to establish the core--\hspace{0pt}Walras equivalence with infinite-dimensional commodity spaces via the use of the separation theorem; see \citet{be73,eh08,ry93}. 

The purpose of this paper is twofold. 

First, we present the equivalence between the fuzzy core and the core under minimal assumptions unlike \citet{no00,bg14}. Due to the exact version of the Lyapunov convexity theorem in Banach spaces established in \citet{ks13}, we clarify that the additional structure of commodity spaces and preferences is unnecessary whenever the measure space of agents is ``saturated'' in the sense of \citet{ks09}. Although \citet{bg14} removed the convexity of preferences from \citet{no00} whenever the measure space of agents is nonatomic, they must impose the requirement mentioned above because of the inevitable use of the approximate version of the Lyapunov convexity theorem; see also \citet{eh08,ry93}.                                          

Second, as a spin-off of the above equivalence, we obtain the coincidence of the core, the fuzzy core, and the restricted core under minimal assumptions. The equivalence between the core and the restricted core in large economies along the lines of \citet{sc72} was extended by \citet{ks13} to Banach commodity spaces under saturation, but the coincidence of the fuzzy core and the restricted core has not been articulated anywhere. As a matter of course, with the aforementioned additional assumption, the classical result on the core--\hspace{0pt}Walras equivalence in large economies along the lines of \citet{au64,hi74} leads to the restricted core--\hspace{0pt}Walras equivalence as well as the coincidence of the fuzzy core and the restricted core in infinite-\hspace{0pt}dimensional commodity space whenever the measure space of agents is nonatomic. 

The organization of the paper is as follows. As preliminaries, Section 2 collects some terminologies on Bochner integrals and vector measures in Banach spaces, provides the definition of saturation of measure spaces, and then presents the infinite-\hspace{0pt}dimensional Lyapunov convexity theorem under saturation. Section 3 explores the main result, the fuzzy core equivalence as is mentioned above.

\section{Preliminaries}
\subsection{Bochner Integrals and Vector Measures}
Let $(T,\Sigma,\mu)$ be a complete finite measure space and $(E,\|\cdot\|)$ be a Banach space. A function $f:T\to E$ is said to be \textit{strongly measurable} if there exists a sequence of simple (or finitely valued) functions $f_n:T\to E$ such that $\|f(t)-f_n(t)\|\to 0$ a.e.\ $t\in T$; $f$ is said to be \textit{Bochner integrable} if it is strongly measurable and $\int_T \|f(t)\|d\mu<\infty$, where the \textit{Bochner integral} of $f$ over $A\in \Sigma$ is defined by $\int_A fd\mu=\lim_n\int_A f_nd\mu$. By the Pettis measurability theorem (see \citet[Theorem II.1.2]{du77}), $f$ is strongly measurable if and only if it is Borel measurable with respect to the norm topology of $E$ whenever $E$ is separable. Denote by $L^1(\mu,E)$ the space of ($\mu$-equivalence classes of) $E$-\hspace{0pt}valued Bochner integrable functions on $T$ such that $\| f(\cdot) \|\in L^1(\mu)$, normed by $\| f \|_1=\int_T \| f(t) \|d\mu$. 

A countably additive set function from $\Sigma$ into $E$ is called a \textit{vector measure}. For a vector measure $m:\Sigma\to E$, a set $N\in \Sigma$ is said to be \textit{$m$-null} if $m(A\cap N)=\mathbf{0}$ for every $A\in \Sigma$. A vector measure $m:\Sigma\to E$ is said to be \textit{$\mu$-\hspace{0pt}continuous} (or \textit{absolutely continuous} with respect to $\mu$) if every $\mu$-null set is $m$-null. For a scalar valued simple function $\varphi$ on $T$ with $\varphi=\sum_{i=1}^n\alpha_i\chi_{A_i}$, where $\alpha_1,\dots,\alpha_n$ are nonzero scalars, $A_1,\dots,A_n\in \Sigma$ are pairwise disjoint, and $\chi_{A_i}$ is the indicator function of $A_i$, define $\Psi_m(\varphi):=\sum_{i=1}^n\alpha_im(A_i)$. Then $\Psi_m$ is a continuous linear operator from the space of simple functions of the above form into $E$; see \citet[p.\,6]{du77}. Thus, $\Psi_m$ has a unique continuous extension (still denoted by $\Psi_m$) to $L^\infty(\mu)$. Hence, it is legitimate to define the integral of $\varphi\in L^\infty(\mu)$ with respect to the vector measure $m$ via $\int_T\varphi(t)d\mu:=\Psi_m(\varphi)$; see \citet[Theorem I.1.13]{du77}. This integral is called the \textit{Bartle integral} of $\varphi$.

\subsection{Lyapunov Convexity Theorem in Banach Spaces}
A finite measure space $(T,\Sigma,\mu)$ is said to be \textit{essentially countably generated} if its $\sigma$-\hspace{0pt}algebra can be generated by a countable number of subsets together with the null sets; $(T,\Sigma,\mu)$ is said to be \textit{essentially uncountably generated} whenever it is not essentially countably generated. Let $\Sigma_S=\{ A\cap S\mid A\in \Sigma \}$ be the $\sigma$-\hspace{0pt}algebra restricted to $S\in \Sigma$. Denote by $L^1_S(\mu)$ the space of $\mu$-\hspace{0pt}integrable functions on the measurable space $(S,\Sigma_S)$ whose elements are restrictions of functions in $L^1(\mu)$ to $S$. An equivalence relation $\sim$ on $\Sigma$ is given by $A\sim B \Longleftrightarrow \mu(A\triangle B)=0$, where $A\triangle B$ is the symmetric difference of $A$ and $B$ in $\Sigma$. The collection of equivalence classes is denoted by $\Sigma(\mu)=\Sigma/\sim$ and its generic element $\widehat{A}$ is the equivalence class of $A\in \Sigma$. We define the metric $\rho$ on $\Sigma(\mu)$ by $\rho(\widehat{A},\widehat{B})=\mu(A\triangle B)$. Then $(\Sigma(\mu),\rho)$ is a complete metric space (see \citet[Lemma 13.13]{ab06}) and $(\Sigma(\mu),\rho)$ is separable if and only if $L^1(\mu)$ is separable; see \citet[Lemma 13.14]{ab06}. The \textit{density} of $(\Sigma(\mu),\rho)$ is the smallest cardinal number of the form $|\U|$, where $\U$ is a dense subset of $\Sigma(\mu)$. 

\begin{dfn}
A finite measure space $(T,\Sigma,\mu)$ is \textit{saturated} if $L^1_S(\mu)$ is nonseparable for every $S\in \Sigma$ with $\mu(S)>0$. 
\end{dfn}

\noindent
The saturation of finite measure spaces is also synonymous with the uncountability of the density of $\Sigma_S(\mu)$ for every $S\in \Sigma$ with $\mu(S)>0$; see \citet[331Y(e) and 365X(p)]{fr12}. Saturation implies nonatomicity; in particular, a finite measure space $(T,\Sigma,\mu)$ is nonatomic if and only if the density of $\Sigma_S(\mu)$ is greater than or equal to $\aleph_0$ for every $S\in \Sigma$ with $\mu(S)>0$. Several equivalent definitions for saturation are known; see \citet{fk02,fr12,hk84,ks09}. One of the simple characterizations of the saturation property is as follows. A finite measure space $(T,\Sigma,\mu)$ is saturated if and only if $(S,\Sigma_S,\mu)$ is essentially uncountably generated for every $S\in \Sigma$ with $\mu(S)>0$. A germinal notion of saturation already appeared in \citet{ka44,ma42}. 

For our purpose, the power of saturation is exemplified in the Lyapunov convexity theorem in infinite dimensions. 

\begin{thm}[\citet{ks13}]
\label{thm1}
Let $(T,\Sigma,\mu)$ be a saturated finite measure space and $E$ be a separable Banach space. If $m:\Sigma\to E$ is a $\mu$-\hspace{0pt}continuous vector measure, then $m$ has weakly compact convex range with: 
$$
m(\Sigma)=\left\{ \int_T \varphi(t)dm\in E\mid 0\le \varphi\le 1,\ \varphi\in L^\infty(\mu) \right\}.
$$
Conversely, every $\mu$-continuous vector measure $m:\Sigma\to E$ has weakly compact convex range, then $(T,\Sigma,\mu)$ is saturated whenever $E$ is infinite dimensional. 
\end{thm}

\begin{rem}
The significance of the saturation property lies in the fact that it is necessary and sufficient for the Lyapunov convexity theorem (see \citet{ks13,ks15,ks16,gp13}), the weak compactness and convexity of the Bochner integral of a multifunction (see \citet{po08,sy08}), the bang-\hspace{0pt}bang principle (see \citet{ks14b,ks16}), and Fatou's lemma (see \citet{ks14a,kss16}). For a further generalization of Theorem \ref{thm1} to nonseparable locally convex spaces, see \citet{gp13,ks15,ks16,sa17,ur19}. Another intriguing characterization of saturation in terms of the existence of Nash equilibria in large games is found in \citet{ks09}. 
\end{rem}

\section{Fuzzy Core Equivalence}
\subsection{Fuzzy Core}
We can now turn to the substantive formulation of a large economy along the lines of \citet{au64,hi74}. Let $(T,\Sigma,\mu)$ be a finite measure space of \textit{agents} with its generic element denoted by $t\in T$. A \textit{commodity space} $E$ is a Banach space. A \textit{consumption set} $X(t)$ of each agent is described by a multifunction $X:T\twoheadrightarrow E$ with $X(t)\subset E$ for every $t\in T$. A \textit{preference relation} ${\succ}(t)\subset X(t)\times X(t)$ of each agent is described by a multifunction $\succ:T\twoheadrightarrow E\times E$ such that ${\succ}(t)$ is a transitive and irreflexive binary relation on the consumption set $X(t)$, where the relation $(x,y)\in {\succ}(t)$ is denoted by $x\,{\succ}(t)\,y$. An \textit{initial endowment} of each agent $\omega(t)\in X(t)$ is provided by a Bochner integrable function $\omega\in L^1(\mu,E)$. An \textit{economy} $\E$ is a quadruple $\E=[(T,\Sigma,\mu),(X(t))_{t\in T}, ({\succ}(t))_{t\in T},(\omega(t))_{t\in T}]$. 

A Bochner integrable function $f\in L^1(\mu,E)$ is called an \textit{assignment} for an economy $\E$ if $f(t)\in X(t)$ a.e.\ $t\in T$. An assignment $f$ is called an \textit{allocation} if $\int_T f(t)d\mu=\int_T \omega(t)d\mu$. A \textit{coalition} is a nonempty set $A$ in $\Sigma$ with $\mu(A)>0$. A coalition $A\in \Sigma$ \textit{blocks} an allocation $f$ if there exists an allocation $g$ such that $g(t){\succ}(t)f(t)$ for every $t\in A$ and $\int_Ag(t)d\mu=\int_A\omega(t)d\mu$. The set of all allocations that no coalition in $\Sigma$ can block is called the \textit{core} of the economy $\E$, denoted by $\C(\E)$.

A measurable function $\alpha:T\to [0,1]$ with $\mu(\{ \alpha>0 \})>0$ is called a \textit{fuzzy coalition}. Following \citet{au81,au82}, $\alpha(t)$ is regarded as the \textit{rate of participation} of agent $t$ in the coalition $\{ \alpha>0 \}\in \Sigma$. If a fuzzy coalition $\alpha$ is an indicator function of the form $\chi_A$, then it is identified with a ``crisp'' coalition $A$. A fuzzy coalition $\alpha:T\to [0,1]$ \textit{blocks} an allocation $f$ if there exists an allocation $g$ such that $g(t){\succ}(t)f(t)$ for every $t\in \{ \alpha>0 \}$ and $\int_T\alpha(t)g(t)d\mu=\int_T\alpha(t)\omega(t)d\mu$. The set of all  allocations that no fuzzy coalition can block is called the \textit{fuzzy core} of the economy $\E$, denoted by $\C^F(\E)$. The inclusion $\C^F(\E)\subset \C(\E)$ follows from the definitions. 

Under the saturation of the measure space of agents, the coincidence of the core and fuzzy core can be established without any order structure on the Banach commodity space and strict monotonicity of the preferences unlike \citet{no00,bg14}. 

\begin{thm}
\label{thm2}
Let $(T,\Sigma,\mu)$ be a saturated finite measure space and $E$ be a Banach space. Then $\C(\E)=\C^F(\E)$.
\end{thm}

\begin{proof}
It suffices to show the inclusion $\C(\E)\subset \C^F(\E)$ is obvious. Toward this end, we first treat the case where $E$ is separable. Suppose to the contrary that there is $f\in \C(\E)$ that does not belong to $\C^F(\E)$. Then there exist a fuzzy coalition $\alpha:T\to [0,1]$ and an allocation $g\in L^1(\mu,E)$ such that $g(t){\succ}(t)f(t)$ for every $t\in \{ \alpha>0 \}$ and $\int_T\alpha(t)g(t)d\mu=\int_T\alpha(t)\omega(t)d\mu$. Define the vector measure $m:\Sigma \to E$ by 
\begin{equation}
\label{eq1}
m(S):=\int_S(g(t)-\omega(t))d\mu, \quad S\in \Sigma. 
\end{equation}
Since $m$ is absolutely continuous with respect to the saturated finite measure $\mu$, in view of Theorem \ref{thm1}, there exists $A\in \Sigma$ such that $m(A)=\int_T\alpha(t)dm$, and hence, $\int_A(g(t)-\omega(t))d\mu=\int_T\alpha(t)dm=\int_T\alpha(t)(g(t)-\omega(t))d\mu=0$. This means that coalition $A$ blocks the allocation $f$ in the core, a contradiction.  

Separability of $E$ can be removed by the following procedure. Note that the vector measure $m:\Sigma \to E$ defined in \eqref{eq1} has a relatively compact range because it has a Bochner integrable density; see \citet[Corollary II.3.9]{du77}. Since the closure $\overline{m(\Sigma)}$ of the range $m(\Sigma)$ is compact, it is also complete and separable; see \citet[Theorem I.6.15]{ds58}. Then take a countable dense subset $D$ of $\overline{m(\Sigma)}$ and consider its closed linear hull $\widetilde{E}:=\overline{\mathrm{span}\,(D)}$ spanned by $D$. Then $\widetilde{E}$ is separable (see \citet[Lemma II.1.5]{ds58}), and hence, it is a closed separable vector subspace of $E$ containing $\overline{m(\Sigma)}$ due to the inclusion $\mathrm{span}\,(\overline{D})\subset \overline{\mathrm{span}\,(D)}$, where $\mathrm{span}\,(\overline{D})$ is the linear hull of the closure $\overline{D}=\overline{m(\Sigma)}$ of $D$. Finally, apply Theorem \ref{thm1} to the vector measure $m:\Sigma \to \widetilde{E}$ which is induced by restricting the range of $m$ in the above argument. 
\end{proof}

\begin{rem}
The notion of fuzzy core explored in this paper is somewhat different from the one in \citet{no00,bg14}. Specifically, the fuzzy core defined here is smaller than the one defined in the above references because they formulate a fuzzy coalition as a ``simple'' measurable function $\alpha:T\to [0,1]$ with $\mu(\{ \alpha>0 \})>0$, so that every fuzzy coalition in their sense takes only ``finite'' values in the closed unit interval. This means that blocking fuzzy coalitions to an allocation in their sense are automatically blocking fuzzy coalitions in our sense. Such a discrepancy of the two notions of fuzzy core, however, disappears in Theorem \ref{thm2} whenever the measure space of agents is saturated. 
\end{rem}

\subsection{Restricted Core}
According to \citet{sc72} followed by \citet{vi72}, we introduce coalitions whose population size is restricted to an arbitrarily small real number $\varepsilon\in (0,\mu(T)]$. A coalition $A\in \Sigma$ is called an \textit{$\varepsilon$-\hspace{0pt}coalition} if $\mu(A)=\varepsilon$. The set of all allocations that no $\varepsilon$-\hspace{0pt}coalition can block is called the \textit{$\varepsilon$-\hspace{0pt}core} of the economy $\E$, denoted by $\C^\varepsilon(\E)$. It follows from the definitions that $\C(\E)\subset \C^\varepsilon(\E)$ and $\C^F(\E)\subset \C^\varepsilon(\E)$ for every $\varepsilon\in (0,\mu(T)]$.

The following result due to \citet{ks13} is an infinite-\hspace{0pt}dimensional analogue of \citet{sc72} under saturation. For completeness, we provide a more direct and simpler proof based on Theorem \ref{thm1}. 

\begin{lem}[\citet{ks13}]
\label{lem1}
Let $(T,\Sigma,\mu)$ be a saturated finite measure space and $E$ be a Banach space. If $\int_Af(t)d\mu=\mathbf{0}$ for some $f\in L^1(\mu,E)$ and $A\in \Sigma$ with $\mu(A)>0$, then for every $\theta\in [0,1]$ there exists $B\in \Sigma$ with $B\subset A$ such that $\int_Bf(t)d\mu=\mathbf{0}$ and $\mu(B)=\theta \mu(A)$.
\end{lem}

\begin{proof}
Assume first that $E$ is separable. Suppose that $\int_Af(t)d\mu=\mathbf{0}$ for some $f\in L^1(\mu,E)$ and $A\in \Sigma$ with $\mu(A)>0$. Define the $\mu$-\hspace{0pt}continuous vector measure $m:\Sigma\to E\times \R$ by 
$$
m(S):=\left( \int_Sf(t)d\mu,\mu(S) \right), \quad S\in \Sigma. 
$$
In view of Theorem \ref{thm1}, the range $m(\Sigma_A)$ is convex. Since $(\mathbf{0},\theta \mu(A))=\theta m(A)+(1-\theta)m(\emptyset)\in m(\Sigma_A)$ for every $\theta\in [0,1]$, there exists $B\in \Sigma_A$ such that $(\mathbf{0},\theta \mu(A))=m(B)=(\mathbf{0},\mu(B))$ by the convexity of $m(\Sigma_A)$. Separability of $E$ can be removed from the above argument by the same procedure as demonstrated in the proof of Theorem \ref{thm1}.
\end{proof}

As demonstrated in \citet{sc72} for finite-dimensional commodity spaces, if an allocation $f$ does not belong to $\C(\E)$, then arbitrarily small $\varepsilon$-\hspace{0pt}coalitions can block $f$. More specifically, we reproduce the following result from \citet{ks13}.

\begin{thm}[\citet{ks13}]
\label{thm3}
Let $(T,\Sigma,\mu)$ be a saturated finite measure space and $E$ be a Banach space. If $f$ is an allocation that is blocked by a coalition $A\in \Sigma$ via an allocation $g$, then for every $\varepsilon\in (0,\mu(A)]$ there exists an $\varepsilon$-coalition $B\in \Sigma$ with $B\subset A$ such that $B$ blocks $f$ via $g$.
\end{thm}

\begin{proof}
Since $A\in \Sigma$ is a blocking coalition to $f$, there exists an allocation $g$ such that $\int_A(g(t)-\omega(t))d\mu=\mathbf{0}$. Take any $\varepsilon\in (0,\mu(A)]$. Applying Lemma \ref{lem1} for the choice of $\theta=\varepsilon/\mu(A)\in (0,1]$ yields that there exists an $\varepsilon$-coalition $B\in \Sigma_A$ such that $\int_B(g(t)-\omega(t))d\mu=\mathbf{0}$. Hence, $f$ is blocked by the $\varepsilon$-\hspace{0pt}coalition $B$ with $B\subset A$ via the  allocation $g$.
\end{proof}

\begin{cor}
\label{thm4}
There exists $\delta\in (0,\mu(T)]$ such that $\C^\varepsilon(\E)=\C(\E)=\C^F(\E)$ for every $\varepsilon\in (0,\delta]$.
\end{cor}

\begin{proof}
In view of Theorem \ref{thm2}, it suffices to show the inclusion for some $\delta\in (0,\mu(T)]$: $\C^\varepsilon(\E)\subset \C(\E)$ for every $\varepsilon\in (0,\delta]$. Suppose to the contrary that for every $\delta\in (0,\mu(T)]$ there exist $\varepsilon\in (0,\delta]$ and $f\in \C^\varepsilon(\E)$ such that $f$ does not belong to $\C(\E)$. Then some coalition $A\in \Sigma$ blocks $f$. Choose here $\delta=\mu(A)$. It follows from Theorem \ref{thm3} that there exists an $\varepsilon$-\hspace{0pt}coalition $B\in \Sigma$ with $B\subset A$ such that $B$ blocks $f$, a contradiction to the fact that $f\in \C^\varepsilon(\E)$ with $\varepsilon\in (0,\mu(A)]$.  
\end{proof}

\section{Concluding Remark} 
This is a sharply focused note on a role of the infinite-dimensional version of the Lyapunov convexity theorem on an equivalence theorem for Aubin's fuzzy core, and we conclude it by drawing attention to the resemblance between Aubin's construct and that of Aumann--Shitovitz on cores of measure-\hspace{0pt}theoretic economies with atoms; see \citet{gs92} for a survey, and \citet{adku21} for ongoing work. Zadeh's insight has now turned into a mature subfield of applied mathematics (see, for example \citet{bdk17} and we think it worthwhile to  pursue this connection that has so far eluded the workers in either register, and  both the disciplines of economics and of applied mathematics.
\clearpage

\bibliographystyle{natbib}

\end{document}